\documentclass[11pt, a4paper]{article}
\usepackage{enumerate}
\usepackage[hyperindex,breaklinks,colorlinks,citecolor=blue]{hyperref}
\usepackage[latin1]{inputenc}
\usepackage{amsmath, amsthm, amsfonts, amssymb}
\usepackage{mathrsfs}
\usepackage{graphicx}
\usepackage{latexsym}
\usepackage{fullpage}
\usepackage[all]{xy}

\newtheorem{proposition}{Proposition}[section]
\newtheorem{theorem}{Theorem}[section]

\theoremstyle{definition}
\newtheorem{definition}{Definition}[section]

\theoremstyle{remark}

\newtheorem{rmks}{Remarks}[section]

\newtheorem*{Examples}{Examples}

\newcommand{\defin}[1]{\emph{\textbf{#1}}}

\newcommand{\C}{\mathscr{C}}

\newcommand{\N}{\mathbb{N}}

\newcommand{\res}[1]{\mathbin\upharpoonright_{#1}}

\title{A Rice-like theorem for primitive recursive functions}
\author{Mathieu Hoyrup}

\begin{document}
\maketitle
\begin{abstract}
We provide an explicit characterization of the properties of primitive recursive functions that are decidable or semi-decidable, given a primitive recursive index for the function. The result is much more general as it applies to any c.e.\ class of total computable functions. This is an analog of Rice and Rice-Shapiro theorem, for restricted classes of total computable functions.
\end{abstract}

\section{Introduction}
Rice theorem \cite{Rice53} states that no non-trivial property of partial computable functions can be decided when the function is presented by one of its indices, or equivalently by a program computing it. Rice-Shapiro theorem \cite{Shapiro56} is a refinement that identifies the properties that are \emph{semi-decidable} when the input function is given by one of its indices, and states that they are exactly the properties that are semi-decidable if the input function is presented by an oracle.

When restricting to the class of total computable functions, Kreisel-Lacombe-Sch\oe nfield \cite{KreLacSch57} and Ceitin \cite{Ceitin62} theorem implies that the decidable properties are the same when presenting the input function via an index or an oracle. However Friedberg \cite{Friedberg58} constructed a set that is semi-decidable when the input function is given by an index, but not when it is given by an oracle. 

In this note we restrict our attention to subclasses of total computable functions, such as the class of primitive recursive functions, the class FP of polynomial-time computable functions, the class of provably total functions in some fixed recursive system (Peano arithmetic, say), etc. The general problem is to understand what can be decided/semi-decided about a function $f$ in such a class if $f$ is given by one of its indices: here an index is essentially a program in a restricted language that only allows to compute functions in the class. The function that decides or semi-decides the property is a usual computable function and is not assumed to belong to the class.

We give a characterization of the decidable and semi-decidable properties, i.e.\ an analog of Rice, Rice-Shapiro and Kreisel-Lacombe-Sch\oe nfield/Ceitin theorems. This characterization uses a version of Kolmogorov complexity restricted to the class of functions under consideration.

\paragraph{Notations.}
Let $\N^*$ be the set of finite sequences of natural numbers and $\N^\N$ the Baire space of functions from $\N$ to $\N$. Given $f\in\N^\N$ and $n\in\N$, $f\res{n}$ denotes the finite sequence $(f(0),\ldots,f(n-1))\in\N^*$. We say that $f\in\N^\N$ extends $v=(v_0,\ldots,v_{n-1})\in\N^*$ if $f(0)=v_0,\ldots,f(n-1)=v_{n-1}$, i.e.\ if $f\res{n}=v$. We denote by $[v]\subseteq\N^\N$ the set of all extensions of $v$ and call it a cylinder. An open set of the Baire space is a union of cylinders. An effective open set is the union of a computable sequence of cylinders.

\section{The result}
Let $\C$ be a class of total computable functions that is computably enumerable: there is a numbering $\C=\{f_i:i\in\N\}$ such that $f_i$ is computable uniformly in $i$, i.e.\ the mapping $(i,n)\mapsto f_i(n)$ is computable. If $f\in\C$ then a $\C$-index of $f$ is any $i$ such that $f=f_i$ (a function may have several indices).

Examples of such classes are: the primitive recursive functions, the class FP of polynomial-time computable functions, the class of provably total computable functions in Peano arithmetic.

First observe that if $f\in\C$ is given by an oracle then a set $A\subseteq\C$ is semi-decidable exactly when $A$ is the intersection of an effective open subset of the Baire space with $\C$; hence $A$ is decidable exactly when both $A$ and $\C\setminus A$ are the intersections of effective open sets with $\C$. Indeed, the machine semi-deciding $f\in A$ accepts $f$ in finite time so it has only read a finite segment of $f$ hence it will accept all functions in some cylinder $[f\res{n}]$.

In order to identify the properties that are decidable or semi-decidable, given a $\C$-index of the input function, we introduce a notion of Kolmogorov complexity adapted to the class $\C$.
\begin{definition}
The $\C$-complexity of $f:\N\to\N$ is
\begin{equation*}
K_\C(f)=\begin{cases}
\min\{i:f_i=f\}&\text{if $f\in\C$,}\\
+\infty&\text{otherwise.}
\end{cases}
\end{equation*}

If $v=(v_0,\ldots,v_n)$ is a finite sequence of natural numbers then its $\C$-complexity is
\begin{equation*}
K_\C(v)=\min\{i:f_i\text{ extends }v\}=\min\{K_\C(f):f\in[v]\}.
\end{equation*}
\end{definition}

Observe that for $f:\N\to\N$, $K_\C(f\res{n})$ is nondecreasing and converges to $K_\C(f)$. By the assumptions on $\C$, the quantity $K_\C(v)$ is computable from $v$, contrary to usual Kolmogorov complexity which is upper semi-computable only. However $K_\C$ usually does not belong to the class $\C$ (modulo encoding of $\N^*$ in $\N$).

It could seem more consistent with usual notions of Kolmogorov complexity (see e.g.\ \cite{Vit93}) to take for instance $\log(i)$ instead of $i$ in the definition, or to use a machine that is universal for the class $\C$ and define $K_\C$ in terms of the size of its inputs. All these choices are equally acceptable and lead exactly to the same result. The important point is that for each such notion of complexity $K'$, an upper bound on $K_\C(f)$ can be uniformly computed from any upper bound on $K'(f)$ and vice-versa. Here we take the simplest definition of complexity following directly from the enumeration of $\C$, to avoid technicality.


First we give a class of decidable properties. Let $h:\N\to\N$ be a computable order, i.e.\ a computable non-decreasing unbounded function.

\begin{definition}
Let $\C$ be an effective class of functions and $h$ a computable order. We define the set $A_{\C,h}$ of \defin{$(\C,h)$-anticomplex} functions as
\begin{equation*}
A_{\C,h}=\{f:\N\to\N:\forall n, K_\C(f\res{n})\leq h(n)\}.
\end{equation*}
\end{definition}

\begin{proposition}\label{prop_A}
For $f\in\C$, the property $f\in A_{\C,h}$ is decidable given any $\C$-index of $f$.
\end{proposition}
\begin{proof}
Given an index $i$ for $f$, one has $K_\C(f\res{n})\leq K_\C(f)\leq i$  for all $n$, so $f$ belongs to $A_{\C,h}$ if and only if $K_\C(f\res{n})\leq h(n)$ for all $n$ such that $h(n)<i$, which is decidable as $K_\C(f\res{n})$ is computable from $i$ and $n$ and only a finite number of values of $n$ has to be checked.
\end{proof}

In general $A_{\C,h}$ is no more decidable if instead of giving an index of $f$ one is only given $f$ as oracle. It contrasts with what happens on the class of partial or total computable functions.

\begin{proposition}
If $\C$ is dense in $\N^\N$ then $A_{\C,h}$ has empty interior in $\C$ (i.e.\ does not contain any intersection of a cylinder with $\C$), therefore $A_{\C,h}$ is not semi-decidable when the input function is given as oracle.
\end{proposition}
\begin{proof}
For each $u=(u_0,\ldots,u_{n-1})\in\N^*$, there exist only finitely many $i\in\N$ such that $K(u_0,\ldots,u_{n-1},i)\leq h(n+1)$. Take any $i$ outside this finite set: the cylinder $[u_0,\ldots,u_{n-1},i]$ is disjoint from $A_{\C,h}$ but intersects $\C$, so $[u]\cap \C$ is not contained in $A_{\C,h}$.
\end{proof}

All the usual classes of total computable functions are dense in $\N^\N$. More generally,
\begin{proposition}
If $\C$ contains a non-isolated function $f$ then there is a computable order $h$ such that $f$ belongs to $A_{\C,h}$ but not to its interior, therefore $A_{\C,h}$ is not semi-decidable when the input function is given as oracle.
\end{proposition}
\begin{proof}
We build $h\geq K_\C(f)$ so that $f\in A_{\C,h}$. As $f$ is not isolated in $\C$, for each $n$ there exist infinitely many functions in $[f\res{n}]\cap \C$. In particular there exist $n+2$\footnote{Here we use the fact that there are at most $n+1$ functions of $\C$-complexity at most $n$. For other notions of $\C$-complexity, this should be adapted, replacing $n+2$ with $2^{n+1}$ for instance.} pairwise distinct functions in that set and a $p$ such that all these functions differ on their prefixes of length $p$. As $\C$ is a c.e.\ class one can compute such a $p$. In this way we can construct an increasing computable sequence $p_n$. For each $n$ there exists $u\in \N^{p_n}$ such that $[u]$ intersects $\C$ and is contained in $[f\res{n}]$ and $K_\C(u)> n$, as there are at most $n+1$ sequences of a given length (here $p_n$) whose complexity is bounded by $n$.

Let $h(p)=\min\{n\geq K_\C(f):p\leq p_n\}$. For each $n\geq K_\C(f)$, $[f\res{n}]\cap\C$ is not contained in $A_{\C,h}$ as there exists $u$ of length $p_n$ such that $[u]$ intersects $\C$ but not $A_{\C,h}$: indeed, $K_\C(u)>n=h(|u|)$. 
\end{proof}

Now observe that the proof of Proposition \ref{prop_A} actually shows that $A_{\C,h}$ remains decidable if one is given $f$ via an oracle \emph{together with an upper bound on $K_\C(f)$}. This is actually the case for every decidable, and even semi-decidable property.

\begin{proposition}\label{prop_K}
Every semi-decidable set $A\subseteq\C$ is semi-decidable given an access to $f$ as oracle and any upper bound on $K_\C(f)$.
\end{proposition}
\begin{proof}
Given $f$ and $k\geq K_\C(f)$, one can progressively reject all numbers $i\leq k$ such that $f_i\neq f$. In parallel one can progressively accept all numbers $i\leq k$ that are accepted by the semi-decision procedure. Wait for a stage when every number $i\leq k$ is accepted or rejected. If this happens then accept $f$.
\end{proof}

It was proved in \cite{STACS15} that such a result also holds for the class of total computable functions and much more general classes of computable objects. The proof given here in the case of an effective class $\C$ is much easier because we only deal with total programs.

We can now state our main result: the cylinders and the sets of anticomplex functions are the basic decidable properties, from which all decidable and semi-decidable properties can be obtained.
\begin{theorem}
Let $A\subseteq \C$. The following conditions are equivalent:
\begin{itemize}
\item $A$ is semi-decidable,
\item $A$ is an effective union of sets of the form $[v]\cap A_{\C,h}$, i.e.
\begin{equation*}
A=\C\cap \bigcup_n([v_n]\cap A_{\C,h_n})
\end{equation*}
for some computable sequences of finite words $v_n\in\N^*$ and orders $h_n:\N\to\N$.
\end{itemize}
\end{theorem}
\begin{proof}
For $k,p\in\N$ let
\begin{align*}
\C_k&=\{f:K_\C(f)\leq k\},\\
\C_k^p&=\{u\in\N^p:K_\C(u)\leq k\}=\bigcup_{f\in \C_k}[f\res{p}].
\end{align*}

By Proposition \ref{prop_K} there exist uniformly effective open sets $U_k\subseteq \N^\N$ such that $A\cap \C_k=U_k\cap\C_k$. One can take $U_{k+1}\subseteq U_k$, replacing $U_k$ with $U_k\cup U_{k+1}\cup\ldots$ if necessary. It follows that $A=\C\cap \bigcap_k U_k$.

Given $k\in \N$ and $v\in\N^*$ such that $[v]\subseteq U_{k+1}$, we now build a computable order $h$ such that $[v]\cap\C_{k+1}\subseteq [v]\cap A_{\C,h}\cap\C\subseteq A$. In order to obtain the announced families $v_n$ and $h_n$ to cover the whole set $A$, we will simply start from all possible $k\in\N$ and all $v$ in the effective enumeration of $U_{k+1}$.

We now define $h$, by first constructing a kind of inverse of $h$. More precisely we define a computable increasing sequence $p_i$ such that for all $i$,
\begin{equation*}
[v]\cap \C_{k+1}^{p_1}\cap\ldots\cap \C_{k+i}^{p_i}\subseteq U_{k+i+1}.
\end{equation*}

The base case $i=0$ is satisfied as $[v]\subseteq U_{k+1}$. Once $p_1,\ldots,p_i$ have been defined,
\begin{align*}
[v]\cap \C_{k+1}^{p_1}\cap \ldots\cap \C_{k+i}^{p_i}\cap \C_{k+i+1}&\subseteq U_{k+i+1}\cap \C_{k+i+1}\\
&\subseteq A\\
&\subseteq U_{k+i+2}.
\end{align*}

The left-hand side is a finite set. For each $f$ in that set, there is $p\in \N$ such that $[f\res{p}]\subseteq U_{k+i+2}$. As the set is finite there is a single $p$ that works for each $f$ in the finite set. As this finite set is computable, such a $p$ can be computed. We then define $p_{i+1}>p_i$ such that
\begin{equation*}
[v]\cap \C_{k+1}^{p_1}\cap \ldots\cap \C_{k+i}^{p_i}\cap \C_{k+i+1}^{p_{i+1}}\subseteq U_{k+i+2}.
\end{equation*}

We then have
\begin{equation*}
[v]\cap \C_{k+1}\subseteq [v]\cap \bigcap_{i\geq 1} \C^{p_i}_{k+i}\cap\C\subseteq \bigcap_n U_n\cap\C=A.
\end{equation*}

Now, $g\in \bigcap_{i\geq 1} \C^{p_i}_{k+i}$ if and only if for all $i$, $K_\C(g\res{p_i})\leq k+i$. Let $h$ be the computable order defined by $h(p)=k+\min\{i\geq 1:p\leq p_i\}$. If $g\in A_{\C,h}$ then $K_\C(g\res{p_i})\leq h(p_i)=k+i$ for all $i$. As $h(p)\geq k+1$ for all $p$, $\C_{k+1}\subseteq A_{\C,h}$.
\end{proof}

A set $A\subseteq\C$ is then decidable from $\C$-indices if and only if both $A$ and $\C\setminus A$ can be expressed as effective unions of sets $[v]\cap A_{\C,h}$.

\end{document}